\documentclass[aps,
pra, 
showpacs,
notfootinbib,
reprint,
10pt,
longbibliography,
superscriptaddress]{revtex4-2}
\pdfoutput=1

\usepackage{asymptote}
\usepackage{amsmath,amsfonts,amssymb,amsthm,bbm,graphicx,enumerate,bm,times}
\usepackage{mathtools}
\usepackage{comment}
\usepackage[table,usenames,dvipsnames]{xcolor}
\usepackage{hyperref}
\PassOptionsToPackage{linktocpage}{hyperref} 

\hypersetup{
  hypertexnames=false,
  colorlinks   = true, 
  urlcolor     = magenta, 
  linkcolor    = blue, 
  citecolor   = green 
}
\usepackage{bookmark} 

\usepackage[capitalize]{cleveref} 

\usepackage{tikz}
\usepackage{graphicx}
\usepackage{bm}
\usetikzlibrary{shapes,positioning}
\usepackage{etoolbox}
\usepackage{todonotes}
\usepackage[linesnumbered,ruled,vlined]{algorithm2e}
\usepackage{siunitx}

 \newtheorem{theorem}{Theorem}

 \newtheorem{lemma}[theorem]{Lemma}

\newcommand{\mc}[1]{\mathcal{#1}}

\newcommand{\Tr}{\mathrm{Tr}} 

\newcommand{\vertiii}[1]{{\left\vert\kern-0.25ex\left\vert\kern-0.25ex\left\vert #1 
    \right\vert\kern-0.25ex\right\vert\kern-0.25ex\right\vert}}

\newcommand{\ZZ}{{\mathbb Z}}

\renewcommand{\vec}[1]{\mathbf{#1}}

\definecolor{edoardo}{rgb}{1,0,0}

\newcommand{\tii}{Quantum Research Centre, Technology Innovation Institute (TII), Abu Dhabi}

\begin{document}

\title{Alleviating the quantum {B}ig-$M$ problem}

\begin{abstract}
A major obstacle for quantum optimizers is the reformulation of constraints as a quadratic unconstrained binary optimization (QUBO).
Current QUBO translators exaggerate the weight $M$ of the penalty terms. Classically known as the ``Big-$M$'' problem, the issue becomes even more daunting for quantum 
solvers, since it affects the physical energy scale. 
We take a systematic, encompassing look at the quantum big-$M$ problem, revealing NP-hardness in finding the optimal $M$ and establishing bounds on
the Hamiltonian spectral gap $\Delta$ as a function of the weight $M$, inversely related to the expected run-time of quantum solvers.
 We propose a practical translation algorithm,
based on SDP relaxation, that outperforms previous methods in numerical benchmarks.
Our algorithm gives values of $\Delta$ orders of magnitude greater, e.g.\ for portfolio optimization instances.
Solving such instances with an adiabatic algorithm on 6-qubits of an IonQ device, we observe significant advantages in time to solution and average solution quality.
Our findings are relevant to quantum and quantum-inspired solvers alike.


\end{abstract}

\author{Edoardo Alessandroni$^*$}
\affiliation{\tii}
\affiliation{SISSA — Scuola Internazionale Superiore di Studi Avanzati, Trieste, Italy}
\email{Corresponding author email: ealessan@sissa.it}
\author{Sergi Ramos-Calderer}
\affiliation{\tii}
\affiliation{Departament de F\'isica Qu\`antica i Astrof\'isica and Institut de Ci\`encies del Cosmos (ICCUB), Universitat de Barcelona, Barcelona, Spain.}
\author{Ingo Roth}
\affiliation{\tii}
\author{Emiliano Traversi}
\affiliation{Department of Information Systems, Data Analytics and Operations, ESSEC Business School, Cergy-Pontoise, France}
\author{Leandro Aolita}
\affiliation{\tii}

\maketitle

\section{Introduction}
Quantum computing holds a great potential for speeding up combinatorial optimization \cite{abbas2023quantum}. 
From a distant-future perspective, the prospects are rooted in the fact that fault-tolerant quantum computers are envisioned to run quantum versions of state-of-the-art classical optimization algorithms more efficiently.
In fact, there is sound theoretical evidence that such quantum algorithms offer a quadratic asymptotic speed-up over their classical counterparts \citep{Montanaro, durr1996quantum,ambainis2019quantum}. 
In the short run, there is a direct relation between the ground state of physical systems and optimizations. The paradigmatic example is Ising models encoding quadratic \emph{unconstrained} binary optimization (QUBO) problems.
This has fueled a quest for ground-state preparation algorithms implementable on nearer-term quantum hardware. These include quantum annealing 
\citep{Farhi00,Albash_AQC,LangZielinskiFeld,Salatino_25,Geier_25,Hegde_2022}, quantum imaginary time evolution
\citep{VariationalQuITE19,motta2020,Nishi21,PW09, chowdhury2017,variationalGibbssampler20, Thais_QITE, Kyaw_2024}, and heuristics such as the quantum approximate optimization algorithm \citep{Farhi_QAOA, basso_et_al:LIPIcs.TQC.2022.7, He_QAOA_PO}. Moreover, apart from quantum solvers, the QUBO paradigm is giving rise to a variety of interesting quantum-inspired solvers as well \citep{Goto_simadiabatic, Kanao_simbifurcation, Mohseni_Isingmach}.

A prerequisite to apply such paradigm to more general quadratically constrained integer optimization problems is to recast them into an equivalent QUBO form.
Recently, automatic QUBO translators have appeared \citep{qiskit_doc, qubovert,zaman2021pyqubo}. 
The translation consists of lifting the constraints to penalty terms in the objective function.
To ensure that the solution to the reformulated (unconstrained) problem coincides with that of the original (constrained) one, the weight of the penalty terms---often denoted as $M$---has to be sufficiently large.
At the same time, choosing an excessively large $M$ causes an increase in run-time due to precision issues in rounding and truncation, even for classical solvers. 
In the classical optimization community this problem is referred to as the \emph{Big-$M$ problem}. 

In contrast, quantum QUBO solvers are closer in spirit to \emph{analog computing devices}. 
There, the value of $M$ directly affects the energy scale of the Hamiltonian whose ground state encodes the solution. 
When performing controlled computations with an actual physical quantum system its admissible energies are limited.
For ground-state preparation schemes these physical limitation eventually also restricts the energy scale of the encoding Hamiltonian \citep{karimi2019practical,Harwood_routing,azad2022solving}.
More precisely, the penalty terms tend to have the undesired side effect of decreasing the spectral gap of the Hamiltonian. 
As a consequence, the precision required to resolve the states, and, hence, also the run-time, increases.
A general rule of thumb is to choose $M$ as small as possible, 
but doing so while still successfully enforcing the constraints is highly non-trivial. 
In fact, the known computationally-efficient Big-$M$ recipes tend to largely over-estimate the required value \citep{Harwood_routing, Leonidas_vehiclerouting, qiskit_doc}. 
Clearly, an efficient QUBO translator with improved spectral properties is highly desirable. Moreover, a formalization of the \emph{quantum big-M problem} as a fundamental concept between quantum physics and computer science is missing too. 
A general framework should address key aspects such as how to quantify the Big-$M$ problem in terms of its impact on quantum solvers or the computational complexity of QUBO reformulations.

\begin{figure*}[t!]
    \includegraphics[width = \textwidth]{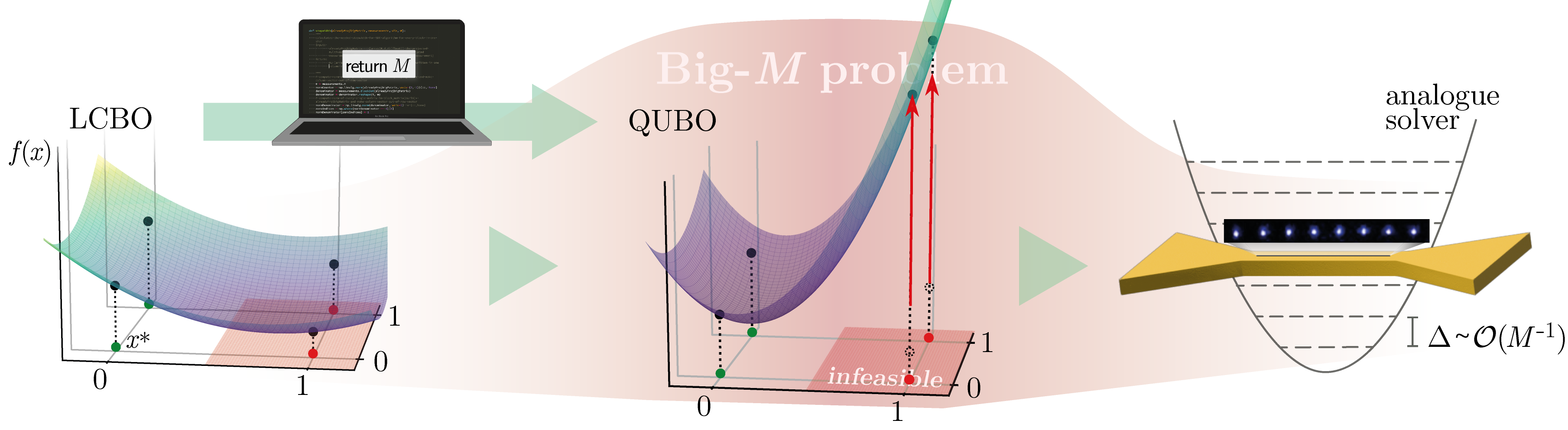}
    \centering
    \caption{
    In a linearly-constrained binary quadratic optimization (LCBO, left) with optimal point $x^\ast$, unfeasible points (red-shaded sub-domain) are excluded from the feasible set (green) via hard constraints.
    For a reformulation as a quadratic unconstrained binary optimization (QUBO, center), the native format for quantum solvers, a penalty term that vanishes on the feasible region with weight $M>0$ is added to the objective function.
    The reformulation is exact ($\vec x^\ast$ remains optimal) if $M$ is sufficiently large, lifting the objective values of all unfeasible points above $f(\vec x^\ast)$.
    However, the penalty term affects the physical energy scale; the larger $M$, the smaller the spectral gap $\Delta$ of the Hamiltonian encoding the reformulation (right). This has a detrimental effect on the runtime of exact solvers, as well as on the quality of the solution of approximate solvers. 
    While NP-hard in general (as we prove), we present an efficient classical strategy to find ‘good’ choices of $M$.
    }
\label{fig:main}
\end{figure*}
Here, we fill in this gap. We develop a theory
of the quantum big-M problem and its impact on the spectral gap $\Delta$.
We start with rigorous definitions for the notions of optimal $M$ and exact QUBO reformulations. 
We prove that finding the optimal $M$ is NP-hard and establish relevant upper
bounds on the Hamiltonian spectral gap $\Delta$, both in terms of the original gap and of $M$. One of these bounds formalizes the intuition that $\Delta=\mathcal{O}(M^{-1})$.
Most importantly, we present a universal QUBO reformulation method with improved spectral properties. 
This is a simple but remarkably-powerful heuristic recipe for $M$, based on standard SDP relaxation \cite{goemans1995improved}. 
We perform exhaustive numerical tests on sparse linearly constrained binary optimizations, set partition problems, and portfolio optimizations from real S\&P 500 data. 
For all three classes, we systematically obtain values of $M$ one order of magnitude smaller and of $\Delta$ from one to two orders of magnitude larger than with state-of-the-art methods, with particularly promising results for portfolio optimization.
In addition, in a proof-of-principle experiment, we solved 6-qubit PO instances with a Trotterized adiabatic algorithm deployed on IonQ's trapped-ion device Aria-1. For the small system size one remains approximately adiabatic with the permissible circuit depth.
We observed that our reformulation increases the probability of measuring the optimal solution by over an order of magnitude and improves the average approximation ratio.
Our findings demonstrate crucial advantages of the proposed optimized QUBO reformulation over the currently known recipes.

\section{Results}
\paragraph*{The quantum Big-$M$ problem---.} Our starting point is a \emph{linearly-constrained binary quadratic optimization (LCBO) problem} 
with $n$ binary decision variables and $m$ constraints,
\begin{equation}\tag{P}\label{lcBQP}
        \operatorname*{minimize}_{\vec x \in \{0,1\}^n}\ f(\vec x) = \vec x^t Q \vec x 
        \quad\operatorname{subject\ to}\quad A\, \vec x = \vec b ,
\end{equation}
specified in terms of $Q \in \ZZ^{n\times n}$, 
$A \in \ZZ^{m\times n}$, and $\vec b \in \ZZ^m$. 
Note that a general polynomially-constrained polynomial optimization problem with integer variables
can always be cast into a linearly-constrained binary quadratic optimization problem of the form \eqref{lcBQP} by standard \emph{gadgets}, as summarized in Supplementary Information \ref{app:gadgets}.
Furthermore, for the sake of clarity, we consider throughout \emph{exact optimization solvers}.  
Clearly, near-term quantum optimization solvers are envisioned to be approximate solvers.
However, our discussion can be extended to approximate solvers too, e.g.\ by considering all admissible approximate solutions as optimal points of problem \eqref{lcBQP}.

To arrive at a QUBO formulation of \eqref{lcBQP}, the best-known strategy (see Fig.~\ref{fig:main}) is to promote the constraints to a quadratic penalty term in the objective function using a suitable constant weight $M>0$. The resulting QUBO reads,
\begin{equation}\tag{P$_{\!M}$}\label{QUBO}
    \operatorname*{minimize}_{\vec x \in \{0,1\}^n}\ \vec x^t Q \vec x + M(A\vec x - \vec b)^2\,.
\end{equation}
We say that \eqref{QUBO} is an \emph{exact reformulation} of \eqref{lcBQP} if their optimal points coincide. 
The penalty term in \eqref{QUBO} vanishes for every feasible point. 
To arrive at an exact reformulation, $M$
has to be chosen large enough for every unfeasible point of \eqref{lcBQP} to have a greater objective value than the original optimum.  
Denoting by $\vec x^\ast$ an optimal point of \eqref{lcBQP}, we have an exact reformulation if and only if there exists a \emph{gap} $\delta > 0$ s.t.\ 
\begin{equation}\label{eq:exact_reformulation condition}
     f(\vec x^\ast) + \delta \leq f(\vec x) + M (A\vec x - \vec b)^2\,
\end{equation}
for all unfeasible points $\vec x$.
There are simple choices of $M$ to ensure this condition, such as 
\begin{equation}\label{eq:Mell1}
    M_{\ell_1} = 
    \|Q\|_{\ell_1} + \delta\,,
\end{equation}
with the
vector $\ell_1$-norm being the sum of all absolute entries. 
Since $M_{\ell_1}$ can be computed in polynomial time, it follows that \eqref{lcBQP} and \eqref{QUBO} with $M = M_{\ell_1}$ are in the same complexity class. 
This choice of $M$ is common \citep{Harwood_routing, Leonidas_vehiclerouting, qiskit_doc}; but, as we show below, it typically yields excessively large values.

We say that a reformulation \eqref{QUBO} of \eqref{lcBQP} has a ($\delta$-)optimal $M$ if it is exact with gap $\delta$ 
and minimal $M$.
Note that the minimal choice of $M$ guarantees only a difference $\delta$ between the optimal objective value and those of the unfeasible points. 
To avoid an arbitrarily small gap, $\delta$ can be chosen as a constant independent on the system size in Eq.~\eqref{eq:exact_reformulation condition}.
For specific classes of problems it is in fact possible to formulate strategies for optimal choices of $M$, 
an example being the problem of finding maximum independent sets, where the optimal value of $M$ is apparent \cite{rydberg}.
In general, however, this is intractable. 
\paragraph*{Observation 1:} Finding an optimal $M$ is NP-hard.

\noindent Intuitively, Eq.~\eqref{eq:exact_reformulation condition} already hints at the possibility that finding the optimal $M$ can be as hard as determining the optimal objective value of the original optimization problem. 
In Sec.~\ref{app:hardness}, we give a polynomial reduction of the problem of deciding if the optimum of $f$ is below a threshold to the problem of deciding if a given $M$ provides an exact reformulation.

From a pragmatic point of view though, it is nonetheless of utmost importance to find suboptimal but `good' choices of $M$ using less resources than required for solving the original problem.
In some specific cases 
(Travelling Salesman Problem \cite{Lucas_ising_form}, permutation problems \cite{ayodele2022penalties}, {e.g.}), 
there are recipes for a `reasonable' value for $M$.
Here, we 
provide a generally applicable strategy to determine `good' choices of $M$.
At the heart of our approach is the following observation.
\paragraph*{Observation 2:} %
\newcommand{\xfeas}{\ensuremath\vec x_\text{\rm feas}}%
\newcommand{\ful}{\ensuremath f_{\widehat{\text{\rm unc}}}}%
Let $\xfeas$ be a feasible point of \eqref{lcBQP}, $\delta > 0$, and choose $\ful \leq \min \{f(\vec x) \mid \vec x \in \{0,1\}^n\}$, i.e.\ as a lower bound on the objective function of \eqref{lcBQP} \emph{when omitting the linear constraints}. Then, \eqref{QUBO} with
\begin{equation}\label{eq:MviaBounds}
    M = f(\xfeas) - \ful + \delta
\end{equation}
is an exact reformulation of \eqref{lcBQP} with gap (at least) $\delta$. 

\paragraph*{Proof.}
Let $\vec x$ be an unfeasible point and $M$ chosen according to Eq.~\eqref{eq:MviaBounds}. 
Since $(A\, \vec x - \vec b)^2 \geq 1$,
$
  M(A\vec x - \vec b)^2 \geq f(\xfeas) - \ful + \delta \geq f(\vec x^\ast) - f(\vec x) + \delta%
$.
The second inequality follows from the definition of $\ful$ and the fact that
$f(\vec x^\ast) \leq f(\xfeas)$ for any feasible point $\xfeas$. 
Thus, Eq.~\eqref{eq:exact_reformulation condition} holds.

While any choice of feasible point and lower bound yields an admissible value of $M$, good choices of $M$ attempt to choose the $\xfeas$ with small objective $f(\xfeas)$ and the bound $\ful$ as tight as possible.
A universal strategy to this end is the following:
i) Find a feasible point $\xfeas$, by running a classical solver on \eqref{lcBQP} limited to some constant amount of time.
ii) Solve the Semi-definite Programming (SDP) relaxation (see Supplementary Information \ref{app:sdp}) of the unconstrained minimization of $f$. Use the resulting objective value as $\ful$. 
This strategy is our main numerical tool. We denote the value given by it as $M_\text{\rm SDP}$.

We note that there exist problem instances where even finding a feasible point is hard. In practice, however, 
there exist efficient heuristics to determine feasible points.
The underlying mindset of our strategy is that modern classical solvers can be powerful allies to quantum optimizers, e.g.\ performing tractable pre-computations to optimize the reformulation for the quantum hardware.
As for concluding 
potential advantages of quantum solvers over classical solvers, with this strategy, one must of course be particularly careful not to accidentally reduce the complexity of the problem in the pre-computation. 
In the case of exact reformulations with optimized $M$, we expect  the complexity  not to decrease even if an optimal $M$ is provided. This expectation is supported by the following analysis.

\begin{figure*}
  \centering
\includegraphics[width=1\linewidth]{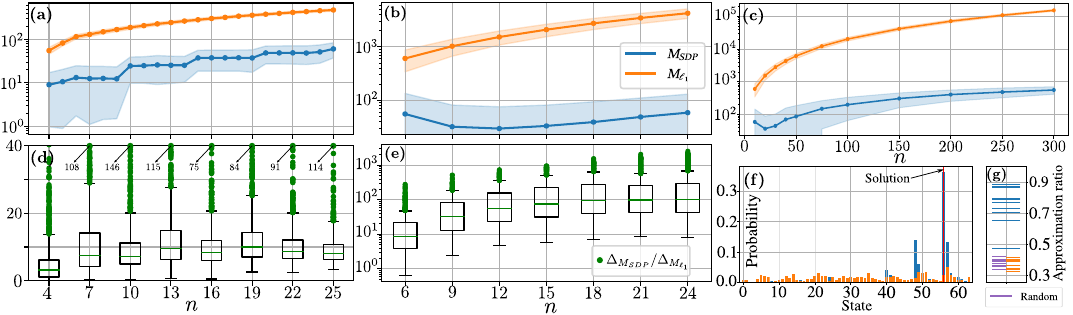}
  \caption{
  {\emph{Numerical results.}} 
  Panels~\textbf{(a)} and \textbf{(b)} depict \emph{numerically} calculated values of $M_\text{\rm SDP}$ (blue) and $M_{\ell_1}$ (orange) for \textbf{(a)} sparse LCBOs (row-sparsity $5$) and \textbf{(b)} portfolio optimization (PO, $w=3$, with $N = n/w$ different stocks up to $8$, $\gamma =1$), averaged over  $1000$ instances for different problem sizes $n$. Shaded stripes indicate the standard deviation. 
  We consistently find significantly smaller values for $M_\text{\rm SDP}$.
  Panels \textbf{(d)} and \textbf{(e)} show box plots of the ratio  $\Delta_{M_\text{\rm SDP}}/\Delta_{M_{\ell_1}}$ of the spectral gaps for the instances in panels \textbf{(a)} and \textbf{(b)}, respectively. Green lines indicate the medians. The whiskers follow the $1.5$ inter-quartile range convention (box contains half of the samples). The black arrows in panel \textbf{(d)} indicate the maximal achieved ratio that is beyond the scope of the displayed $y$-axis. The number at the arrows tip is the maximum value of these outliers. 
The spectral gaps of formulations with  $M_\text{\rm SDP}$ are larger by factors of up to $100$ for sparse LCBOs. For PO the factors increase with $n$, reaching $1000$ for some instances.
  Panel \textbf{(c)} displays values for $M_\text{\rm SDP}$ (using a greedy heuristics) and $M_{\ell_1}$ for larger PO instances ($w = 5$, comprising up to $N=60$ stocks, $\gamma = 1$) averaged over $100$ instances. Shading indicates standard deviation. The gap between $M_\text{\rm SDP}$ and $M_{\ell_1}$ grows with the system size. 
  \textbf{(f)} \emph{Experimental results} of $10$ randomly-selected 6-qubit PO instances for a Trotterized adiabatic solver limited to $150$ two-qubit gates on a trapped-ion IonQ quantum computer. Histogram of $1000$ measurement outcomes of one of the $10$ instances for the $M_\text{\rm SDP}$ and $M_{\ell_1}$ reformulations. The optimal solution (red line) has probability larger than $0.3$ for the $M_\text{\rm SDP}$ reformulation. In contrast, using $M_{\ell_1}$ the experimental distribution is close to uniform. 
  Panel \textbf{(g)} displays the {experimentally obtained} average approximation ratios for the $10$ instances using both reformulations and when sampling uniformly at random from the feasible solutions. 
  With the limited quantum resources of a state-of-the-art noisy quantum device, only instances using $M_\text{\rm SDP}$  outperform the random classical strategy.
  }
  \label{Result_plots}
\end{figure*}

\paragraph*{Spectral gap as a measure of the Big-$M$ problem---.} Near-term quantum (and quantum-inspired) solvers are based on ground-state optimization of an Ising Hamiltonian $H_M:=H_\text{f}+M\, H_\text{c}$ (see Sec.~\ref{app:experimental} for explicit expressions) that encodes the objective function in \eqref{QUBO}. The Hamiltonian $H_\text{f}$ encodes the objective function $f(\vec x)$ of the original problem, while $H_\text{c}$ encodes the constraint term $(A\vec x - \vec b)^2$. Hence, the choice of $M$ directly affects the spectral gap 
\begin{equation}
\Delta_M \coloneqq \frac{E_1-E_0}{E_\text{max}-E_0},
\end{equation}
where $E_0$, $E_1$, and $E_\text{max}$ are respectively the lowest, next-to-lowest, and maximum energies of $H_M$. 
The spectral gap normalization is imperative to realistically compare different Hamiltonian. A physical quantum solver has access to a restricted energy scale that the problem Hamiltonian we want to solve must accommodate to. From a dual perspective, this energy scale relates to numerical precision, as it evaluates the ratio between the accuracy needed to discriminate the lowest-energy state, and the full energy spectrum.

The connection between the spectral gap and $M$ is made apparent through the following calculations (Sec.~\ref{app:boundsdelta}).
\paragraph*{Observation 3:}
If \eqref{QUBO} is an exact reformulation, then
$i$) $E_0=f(\vec x^\ast)$; 
$ii$) $\Delta_M \leq \Delta_0$ for all $M$, with $\Delta_0$ the `spectral gap' of the constrained optimization problem \eqref{lcBQP}; 
and $iii$) $\Delta_M /\Delta_0 \leq (E^{(\text{f})}_\text{max} - E_0) / (E^{(\text{c})}_\text{max} M - M^\ast)$ where $M^\ast$ is an optimal $M$, as defined in the previous section, and $E^{(\text{c})}_\text{max}$ and $E^{(\text{f})}_\text{max}$ are the maximum energies of $H_\text{c}$ and $H_\text{f}$, respectively.
The first implication states simply that $E_0$ equals the optimal objective value of problem \eqref{lcBQP}. The second one that no exact reformulation \eqref{QUBO} can increase the spectral gap. Finally, since both $H_\text{f}$ and $H_\text{c}$ are independent of $M$, the bound in $iii$) implies that $\Delta_M \in \mathcal{O}(\Delta_0/M)$ asymptotically. 
This substantiates the initial intuition that an excessively high penalty $M$ is detrimental to analogue solvers.
For instance, in quantum annealers, adiabaticity requires a run-time $\Omega\big(\Delta_M^{-2}\big)$ \citep{Farhi00,Albash_AQC}. In turn, for imaginary time evolution, the inverse temperature required for constant-error ground-state approximation is $\Omega\big(\Delta_M^{-1}\big)$ \citep{VariationalQuITE19,motta2020,Nishi21,PW09, chowdhury2017,variationalGibbssampler20, Thais_QITE, Kyaw_2024}. 
For variational algorithms, 
also the training is affected: as $M$ grows, the sensitivity of the cost function (the energy) to parameter changes becomes increasingly dominated by $H_\text{c}$ and $H_\text{f}$ increasingly irrelevant \citep{Harwood_routing, Leonidas_vehiclerouting}.
Hence, we propose $\Delta_M$ as a natural measure for the Big-$M$ problem of a QUBO reformulation. 
This allows us to quantitatively benchmark our Big-$M$ recipe against the previous direct bounds, which we do next. 

\paragraph*{Numerical benchmarks---.}
We evaluate the performance of reformulations with optimized $M_\text{\rm SDP}$, against the common choice $M_{\ell_1}$ for three examples of LCBO problem classes:
Random sparse LCBOs, set partitioning problems (SPPs), and portfolio optimization (PO).
Details on the model definitions and further results are presented in Sec.~\ref{app:models}.  

For PO we use the well-known Markovitz model \cite{Markowitz_PO, grant2021markovitz, rosenberg2016risk}, i.e.\ the problem of selecting a set of assets maximizing returns while minimizing risk.
The problem specification requires a vector $\boldsymbol \mu$ of expected returns of a set of $N$ assets, their covariance matrix $\Sigma$, a risk aversion $\gamma > 0$, and a partition number $w$ defining the portfolio discretization.  
Denoting by $x_i$ the units of asset $i$ in the portfolio, the problem formulation reads
\begin{equation}
    \operatorname*{minimize}_{\vec x \in \mathbb N^N}\ -\boldsymbol{\mu}^t \vec x + \gamma\, \mathbf{x}^T \Sigma \mathbf{x} \quad 
\operatorname{subject\,to}\quad  \textstyle\sum_i x_i = 2^w - 1\,. \label{Markowitz_problem}
\end{equation}
The constraint forces the total budget to be  invested. 
The QUBO reduction requires mapping each integer decision variable into $w$ binary variables. 
We generate problem instances from historic financial data on S\&P 500 stocks.

We observe that the result $M_\text{\rm SDP}$ of our algorithm is consistently one order of magnitude smaller than $M_{\ell_1}$ 
for random sparse LCBOs%
~[Fig.~\ref{Result_plots} (a)]
and SPPs [Fig.~\ref{fig:SPP} in Sec.~\ref{app:models}]. 
Concomitantly, the spectral gap is relatively increased by an order of magnitude [Fig.~\ref{Result_plots} (d)].
This corroborates our theoretic consideration relating optimized choices of $M$ to significantly improved spectral properties of the QUBO formulation.
In the PO instances we additionally observe that the advantage in $M$ and $\Delta$ further grows with the problem size [Fig.~\ref{Result_plots} (b) and (e)].
In contrast to random LCBOs and SPPs, PO has a single constraint independent of the problem size. 
While this allows for highly optimized choices of the penalty weight as exemplified by $M_\text{\rm SDP}$, the bound in Eq.~\eqref{eq:Mell1} is oblivious to the intrinsic structure of PO. 
Using a greedy heuristic to determine $f(\ensuremath\vec x_\text{\rm feas})$ (Sec.~\ref{app:greedy}), we further calculate 
$M_\text{\rm SDP}$ for PO instances with up to $300$ binary variables and find that 
the improvements over $M_{\ell_1}$
persist [Fig.~\ref{Result_plots} (c)].

\paragraph*{Quantum hardware deployment---.}
Finally, after building our theoretical framework and numerical methods, we turn to  the question of  
how relevant the Big-$M$ problem on 
actual noisy near-term hardware. 
To this end, we deployed $10$ toy instances of PO on the experimental $25$-qubit trapped-ion quantum computer \emph{Aria-1}  \cite{ionq}. 
As an approximate solver, we executed a Trotterized adiabatic evolution to the Hamiltonians encoding QUBO reformulations (see Sec.~\ref{app:experimental} for implementation details).
We executed a set of $10$ random six-qubit instances with a fixed budget of $150$ two-qubit gates for reformulations with $M_\text{\rm SDP}$ and $M_{\ell_1}$.
The limit on the circuit size and a suitably chosen maximal evolution time determine the number of Trotterization steps.  
The parameter choice ensures approximate adiabaticity.
We find that the probability of measuring the optimal solution is more than an order of magnitude higher with the $M_\text{\rm SDP}$ reformulation than with $M_{\ell_1}$ [Fig.~\ref{Result_plots} (f)]. 
The probability of measuring the optimal solution determines the required number of repetitions and, thus, enters inversely into the time-to-solution.
This behavior is consistently observed across all instances.
Fig.~\ref{Result_plots} (g) shows the average  approximation ratio---quantifying the quality of an approximate solution---%
over all measured outcomes that satisfy the budget constraint per instance.
The $M_\text{\rm SDP}$ formulations yield high  ratios for most instances while the $M_{\ell_1}$ formulations perform comparable to classical uniform random sampling of solutions. 
Thus, already for small instances, we find that using an optimized $M$ is an essential prerequisite for deployment on noisy near-term hardware. 
Given the scaling observed in the numerical benchmarks, we expect that small values of $M$ are even more important for the performances that are not dominated by noise  on intermediate sizes hardware.
Similar results have been observed in simulations of Trotterized adiabatic evolutions; see Fig.~\ref{fig:simulation_trotter} and Supplementary Information \ref{app:simulation}.

\section{Discussion}

On the conceptual side, we formalized the quantum big-$M$ problem, giving rigorous definitions, establishing its computational complexity, and giving bounds on the impact of $M$ on the spectral gap of a QUBO Hamiltonian.
The latter relates the big-$M$ problem to performance guarantees of different solvers. 
From a practitioner's viewpoint, 
our main contribution is a versatile QUBO reformulation algorithm with enhanced spectral properties,
based on the SDP relaxation. Our mindset is that classical solvers should be leveraged to pre-condition problems so as to exploit quantum hardware to its maximal potential---near-term devices in particular. 

In numerical benchmarks, including Markovitz portfolio optimization (PO) instances from  S\&P 500 data, we consistently observe  significant improvements in $M$ and the spectral gap using the proposed algorithm.
In a six-qubit proof-of-principle experiment with trapped ions, we find that these improvements translate into a tangible advantage in the probability of measuring the correct solution as well as the average solution quality. This being already present for small instances, the results presented in Fig.~\ref{fig:main} support that the improvement due to a tighter reformulation will persist for larger instances.

Even when other big-$M$ recipes are available, our method can assist such schemes, e.g.\ by providing a suitable starting point with tractable classical resources. 
For example, using multiple calls to the solver, one can determine a good value for $M$ via binary search. 
After every call, $M$  is increased when the returned solution is infeasible and otherwise decreased in smaller and smaller steps. Starting such a search from an $M$ determined with our method still reduces the number of calls to a potentially expensive solver.

Beyond near-term quantum devices, our analysis of the Big-$M$ problem also applies to future fault-tolerant quantum hardware, for instance in adiabatic schemes \citep{Farhi00,Albash_AQC,Geier_25} or quantum imaginary-time evolution simulations \citep{VariationalQuITE19,motta2020,Nishi21,PW09, chowdhury2017,variationalGibbssampler20, Thais_QITE, Kyaw_2024}. Besides, a particularly interesting question to explore is how beneficial our general big-$M$ recipe is for quantum-inspired, classical solvers \citep{Goto_simadiabatic, Kanao_simbifurcation, Mohseni_Isingmach}.
Combining our approach with modern scalable randomized algorithms for SDPs \cite{SketchySDP} can potentially further reduce the complexity of calculating optimized values for $M$.
Finally, while our method was conceived for general instances, nuances of specific problems can enable heuristic tools for tighter lower bounds or feasible points, 
as already exemplified with the greedy heuristic for PO instances.

\begin{figure*}
  \centering
    \includegraphics[width=1\linewidth]{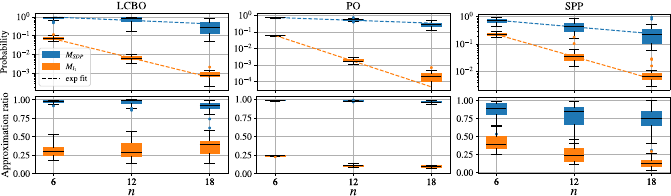}
  \caption{ 
  \emph{Simulation of Trotterized adiabatic evolution.}
  The first and second rows show boxplots of, respectively, the probability of sampling the exact solution and the approximation ratio of the obtained state at the end of a simulation of a trotterized adiabatic-evolution (see Supplementary Information \ref{app:simulation} for details) for different problem sizes $n$ and  QUBO reformulation using penalty weight $M_\text{\rm SDP}$ (blue) and $M_{\ell_1}$ (orange).
  We use $25$ random instances of one of the benchmark models (LCBO, PO, SPP) in each column. The whiskers follow the $1.5$ inter-quartile range convention, the box contains half of the samples. In all settings, we observed that $M_\text{SDP}$ reformulation leads to significantly high probability of obtaining the exact ground state and high approximation ratios, contrarily to $M_{\ell_1}$. In the first row, the dotted lines represent exponential fits to the median of the probability  with the form $p(n) = \alpha e^{-\beta n}$. The fitted decay factors $\beta_{\text{fit}}$ are $0.07(4)$ and $0.390(4)$ for LCBO (blue and orange curves, respectively), $0.07(2)$ and $0.60(2)$ for PO, and $0.093(7)$ and $0.310(5)$ for SPP. 
  The number in parentheses represents the standard deviation in the last digits of the corresponding value.
  We observe that also the ratio of the probabilities for $M_{\ell_1}$ over $M_\text{\rm SDP}$ decreases exponentially with $n$, proportional to $c^{-n}$ with $c$ between $1.24$ and $1.70$ in our examples. Thus, we find a significantly increasing advantage  in, e.g. sampling complexity, with the system size of the optimized penalization strategy over the baseline.
  }
 \label{fig:simulation_trotter}
\end{figure*}

\section{Methods}
\subsection{On the hardness of the quantum Big-$M$ problem.} \label{app:hardness}
\def\decidePM{\textsf{decidePM}}
\def\decideF{\textsf{decideF}}
Here we formally establish that determining an optimal $M$ is in general as hard as finding the objective value of the original problem. {We do this by proving a simple reduction from the decision-problem version of the latter to that of the former.} 

Finding the optimal objective value of a function $f$ under constraints $\mc C$ is equivalent to an associated decision problem, \decideF, {which,} 
given a threshold $a$ and a gap $\delta >0${, decides} if $\min_{{\vec x}\in \mc C} f(\vec x) \leq a$ (`\texttt{smaller}') or if $\min_{{\vec x}\in \mc C} f(\vec x) \geq a + \delta$ (`\texttt{greater}'). 
Note that having access to an oracle for \decideF\ allows one to efficiently find the optimal objective value via binary search.
We want to relate the complexity of \decideF\ to the following decision problem:
\emph{Given an instance of \eqref{lcBQP}, an $M${,} and a gap $\delta > 0$, decide if \eqref{QUBO} with {the given} $M$ is an exact reformulation of \eqref{lcBQP} with gap at least $\delta$ (`\texttt{yes}') or \eqref{QUBO} fails to be an exact reformulation (`\texttt{no}').} We refer to this problem as \decidePM. 
Note that \decidePM\ is equivalent to the problem of finding the $\delta$-optimal $M$: Given an optimal value for $M$, \decidePM\ can be solved by comparing the $M$ under scrutiny to the optimal one. In turn, with  
an oracle for \decidePM, the optimal $M$ can be found via binary search. Next, we prove the promised reduction.

\begin{lemma}\label{lem:PMhardness}
 The problem \decideF\ reduces to \decidePM.
\end{lemma}

\begin{proof}
  Consider an instance $(f, \{0,1\}^n, a, \delta
  )$ of \decideF. 
  W.l.o.g., we assume that the instance is unconstrained. 
  For the constraint problem there exist a polynomial reduction to an unconstrained problem,\ e.g.\ using \eqref{QUBO} with the value of $M$ defined in Eq.~\eqref{eq:Mell1}. 

  We will split \decideF\ into decision problems where we decide the optimum for the subset with constant hamming weight $|\vec x| = \sum_{i=1}^n x_i = k \in \{0, 1, \ldots, n\}$. Deciding for all $k \in \{0, 1, \ldots, n\}$ individually if $\min_{|\vec x|=k} f(\vec x) \leq a$ (`\texttt{smaller}') or $\min_{|\vec x|=k} f(\vec x) \geq a + \delta$ (`\texttt{greater}') allows us to solve \decideF\ in the following way. If all constant-Hamming weight decisions return `\texttt{greater}', we also conclude `\texttt{greater}' for \decideF. If at least one constant-Hamming weight decision returns `\texttt{smaller}', we  return `\texttt{smaller}' for \decideF. 
  It is straight-forward to see that this strategy solves \decideF\ correctly in both cases. 

  It remains to reduce the decision problem with constant Hamming weight $k$ to \decidePM. For $k =0$, i.e.\ $\vec x = 0$, we can directly solve the decision problem by evaluation. If we find $f(\vec 0) \leq a$, we  conclude `\texttt{smaller}' for \decideF.
  Thus, we can restrict our focus in the remainder to $k > 0$ and assume that $f(\vec 0) > a$, where \decideF\ is not yet decided. 
  We choose $\alpha \geq \max_{\vec x\in \{0,1\}^n} f(\vec x) - \min_{\vec x\in \{0,1\}^n} f(\vec x) + f(\vec 0) - a$, e.g.\ $\alpha = M_{\ell_1} + f(\vec 0) - a$ using $M_{\ell_1}$ defined in Eq.~\eqref{eq:Mell1} for the quadratic form $Q$ defining $f$. 
  Consider the following optimization problem:
  \begin{equation}
  \label{eq:trivial_P}
    \operatorname*{minimize}_{\vec x \in \{0,1\}^n} f(\vec x) + \alpha |\vec x| \left| |\vec x| - k\right| 
    \quad\operatorname{subject\ to}\quad \vec x = \vec 0\,.
  \end{equation}
  The optimal point of \eqref{eq:trivial_P} is the only feasible point $\vec x^\ast = \vec 0$ with objective value $f(\vec 0)$. In other words, the constraint renders the optimization problem \eqref{eq:trivial_P} trivial.
  Still of interest to us is the associated problem of deciding if certain values of $M$ yield unconstrained reformulations of \eqref{eq:trivial_P}.
  As formulated \eqref{eq:trivial_P} is not an instance of \eqref{lcBQP}, since
  the objective function is not quadratic.
  But the optimization problem \eqref{eq:trivial_P} can be recast as the following binary quadratic problem: 
  \begin{equation}\begin{split}
  \label{eq:trivial_P_bqp}
    \operatorname*{minimize}_{{\vec x \in \{0,1\}^n}\!\!,\ {p,  m \in \{0,\ldots, n\}}} \hspace{-1.9em} f(\vec x) + g(\vec x, p, m)
    \; \operatorname{subject\ to} \; \vec x = \vec 0\,.
  \end{split}\end{equation}
    with 
    $
        g(\vec x, p, m) = \alpha |\vec x| (p + m) 
    + \alpha ( n^3 + 1) (p - m - |\vec x |  + k)^2
    $.\textbf{}
  Here the non-negative integer variables $p$ and $m$ can each be encoded with $\lceil \log n\rceil$ binary variables. The last summand in $g$ dominates the objective function for all values of $\vec x$, $p$ and $m$.
  Thus, at optimal $p$ and $m$, it 
  enforces the
  constraint $p - m = |\vec x |  - k$. 
  For $\vec x \neq 0$ the minimum of the objective function over $p$ and $m$ is attained when either $p$ or $m$ is equal to $\left||\vec x |  - k\right|$ while the other variable vanishes.   
  We conclude that for all $\vec x$ the objective functions of \eqref{eq:trivial_P} and \eqref{eq:trivial_P_bqp} at optimal $p$ and $m$ coincide.
  Since \eqref{eq:trivial_P_bqp} is an instance of \eqref{lcBQP}, it defines instances of \decidePM. 
  
  We now decide if \eqref{QUBO} with $M = (f(\vec 0) - a)/k > 0$ is an exact reformulation with gap $\delta$ of \eqref{eq:trivial_P_bqp}.  
{}  If the answer of \decidePM\ is `\texttt{yes}' (`\texttt{no}'), we return `\texttt{greater}' (`\texttt{smaller}'). Our claim is  that this strategy correctly solves the decision problem for the minimum of $f$ with constant Hamming-weight $k$.

  To see this, let us first consider the case `\texttt{greater}',  where $\min_{|\vec x| = k} f(\vec x) \geq a + \delta$.   By our choice of $\alpha$, the QUBO reformulation of \eqref{eq:trivial_P_bqp} with $M = (f(\vec 0) - a)/k > 0$ has an objective function that attains its minimum over the unfeasible points for $k = |\vec x|$. Thus, this minimum fulfills  
  \begin{equation} \label{eq:proof:eq1}
  \begin{split}
    &\min_{\vec x \neq 0} \left\{ f(\vec x) + \alpha |\vec x| \left | |\vec x | - k \right | + 
    \frac{f(\vec 0) - a} k |\vec x | \right\} \\
    &\quad= \min_{|\vec x| = k} f(\vec x) + f(\vec 0) - a 
    \geq \delta + f(\vec 0)\,. 
  \end{split}
  \end{equation}
  Due to the trivializing constraint, $f(\vec 0)$ is the optimal value of \eqref{eq:trivial_P_bqp}. 
  Hence, \eqref{eq:proof:eq1} establishes the criterion Eq.~\eqref{eq:exact_reformulation condition} for an exact reformulation.  As required in this case, \decidePM, thus, returns `\texttt{yes}' and we decide correctly.

  Second, let us consider the case `\texttt{smaller}', i.e.\ $\min_{|x| = k} f(\vec x) \leq a$. By the same argument as before, we now find that the minimum of the objective function of the QUBO reformulation over the unfeasible points is smaller or equal than $f(\vec 0)$. Thus, \decidePM\ returns `\texttt{no}' in this case.  Using \decidePM, we therefore always arrive at the correct decision about the minimum of $f$ for constant Hamming-weight.
\end{proof}
Since $\decideF$ encompasses \textsf{NP}-complete problems like \textsf{3SAT}, as a corollary of Lemma~\ref{lem:PMhardness}, we establish that finding the optimal value of $M$ is \textsf{NP}-hard.

\subsection{Bounds on the spectral gap of Big-$M$ QUBO reformulations.}\label{app:boundsdelta}
Next, we provide the detailed argument for \emph{Observation 3} of the main text, and expand on some of its implications. 

  Let $H_M = H_f + M H_\text{\rm c}$ be a Hamiltonian encoding of \eqref{QUBO}.  
  The normalized spectral gap of $H_M$ is defined as 
  \begin{equation}
    \Delta_M \coloneqq \frac{E_1 - E_0}{E_{\text{\rm max}} - E_0}\,, 
  \end{equation}
  where $E_0$, $E_1$ and $E_\text{max}$ are the respective lowest, next-to-lowest and maximum energies of $H_M$. 
  We will study the behavior of $\Delta_M$ compared to the corresponding quantity of the constraint optimization problem \eqref{lcBQP}. 
  To this end, let
  $\vec x^\ast$ be an optimal point as before and 
   let further $\vec x^\ast_1$ be a next-to-optimal point of \eqref{lcBQP}, i.e.\ the optimal point of \eqref{lcBQP} with the additional constraint $\vec x \notin f^{-1}(f(\vec x^\ast))$. 
  Denote by $\mathcal C \coloneqq \{\vec x \mid A\, \vec x = \vec b\}$ the constraint set and by $\overline{\mathcal C} \coloneqq \{0,1\}^n \setminus \mathcal C$ its complement. 
  We define the two upper bounds of the shifted objective function $\bar f \coloneqq \max_{\vec x \in \mathcal C} f(\vec x) - f(\vec x^\ast)$ and $\bar f_c \coloneqq \max_{\vec x \in \overline{\mathcal C}} f(\vec x)  - f(\vec x^\ast)$. We refer to 
  \begin{equation}
    \Delta_0 \coloneqq 
    \frac{f(\vec x^\ast_1) - f(\vec x^\ast)}
    {\bar f}\,
  \end{equation}
  as the \emph{spectral gap} of \eqref{lcBQP}. 

  (\textit{i}) The Ising encoding ensures that $\langle \vec x| H_f | \vec x \rangle = f(\vec x)$ for all $\vec x$, where $|\vec x\rangle$ denotes the basis vector that encodes the binary vector $\vec x$. For $\vec x$ feasible, $|\vec x \rangle$ is in the kernel of $H_\text{\rm c}$. Thus, 
  when $M$ is chosen such that \eqref{QUBO} is an exact reformulation of \eqref{lcBQP}, we have 
  $E_0 = f(\vec x^\ast)$. 

  (\textit{ii}) Analogously, $f(\vec x^\ast_1)$ is still in the spectrum of $H_M$. Thus, $E_1 \leq f(\vec x^\ast_1)$. 
  By Eq.~\eqref{eq:exact_reformulation condition} and assuming $\delta \leq \delta^\ast \coloneqq f(\vec x^\ast_1) - f(\vec x^\ast)$, we have 
  $f(\vec x^\ast) + \delta \leq E_1$. 
  We conclude that $\delta \leq E_1 - E_0 \leq \delta^\ast $ (with the lower bound holding as long as $\delta$ does not exceed the upper bound). 
  Note that the lower bound is saturated for the optimal $M$. 
  Also $\bar f + f(\vec x^\ast)$ is still in the spectrum of $H_M$. 
  Hence, $\bar f \leq E_\text{\rm max} - f(\vec x^\ast)$. 
  All together, combined with (\textit{i}) and its assumption, we arrive at the bound 
  $\Delta_M = (E_1 - E_0) / (E_\text{\rm max} - E_0) \leq \delta^\ast / \bar f = \Delta_0$. 

  (\textit{iii}) To infer the scaling of the spectral gap with $M$, we note that 
  $E_\text{\rm max} - E_0  \geq \min_{\vec x \in \overline{\mathcal C}} \{ f(\vec x) - f(\vec x^\ast)\} + M\, \|H_c\| \geq M\, \|H_c\| - M^\ast$, where we used the positivity of $H_c$ on the infeasible subspace and denote by $M^\ast$ the optimal $M$, i.e.\ the minimal $M$ satisfying Eq.~\eqref{eq:exact_reformulation condition}.  
  Hence, $\Delta_0 / \Delta_M \geq (E_\text{\rm max} - E_0) / \bar f  \geq  (M\|H_c\| - M^\ast) / \bar f \geq  (M\|H_c\| - M^\ast) / \|H_f -E_0\| $.  Thus, in particular $\Delta_M \in O(\Delta_0 / M)$. 

\begin{figure*}[tb]
\includegraphics[width = .8\textwidth]{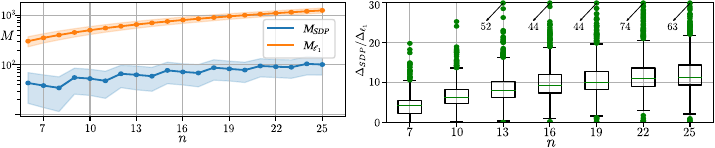}
\centering
\caption{
{\emph{Set Partitioning Problem numerics.}}
Big-$M$ value (left) with two strategies, $M_{\ell_1}$ and $M_\text{\rm SDP}$,  and ratio between spectral gaps resulting from the two choices (right), on a dataset of $1000$ SPP instances generated with density $= 0.25$. The black arrows in the panel on the right indicate the maximal ratio  achieved by extreme outliers outside of the axis' scope. }
\label{fig:SPP}
\end{figure*}

\subsection{Benchmarked models.} \label{app:models}
The present section illustrates the model definition and relevant details of the optimization problems tested, together with further results.

Random sparse LCBOs.~A general class of linearly-constrained binary quadratic optimization problems, whose formulation is \eqref{lcBQP}, have been generated. We choose random instances for $Q$ and $A$ with a bounded row-sparsity $s$, i.e.\ $|\{j : Q_{ij} \ne 0\}| \le s \;\; \forall i $, and similarly for $A$. The non-vanishing entries of $Q$, $A$ and $b$ are uniformly drawn at random. 
We let the number of constraints $m$ grow linearly with the number of binary variables $n$, specifically, $m = \max\{\lfloor \frac{n}{5}\rfloor, 1\}$.

Set partitioning problem (SPP). Let $R_i \subset S$ be a subset of $S=\{1,\dots,m\}$ with an associated cost $c_i\ge 0$, for $i=1,\dots,n$. A family of subsets $\{R_i\}_{i \in W}$ is a partition of $S$ if $\cup_{i \in W} R_i = S$ and $R_i \cap R_j = \emptyset$ for all $i\ne j \in W$.
The SPP consists of finding a partition of $S$ with minimal total cost:
\begin{equation}
\label{eq:SPP}
   \operatorname*{minimize}_{\mathbf{x} \in \{0,1\}^n}
   \; \vec c^t \vec x \quad \operatorname{subject\,to} \quad \textstyle\sum_{i : \alpha \in R_i} x_i = 1 \;\forall\ \alpha\ \in S\,.
\end{equation}

The objective variables encode the subset family, with $x_i = 1$ if $R_i$ is in the family and $0$ otherwise.
The constraints force the family to be a partition of $S$. 
We generate instances by randomly selecting constraint matrices $A$ with 
fixed density, i.e.\ number of non-zero entries over total number of entries.
Fig.~\ref{fig:SPP} shows simulations results relative to this problem class.

Portfolio Optimization (PO). The present paragraph illustrates how the data used in Portfolio Optimization instances were fetched from real data and adapted to Markowitz formulation \eqref{Markowitz_problem}. From stock market index S\&P500, we downloaded the stock price history, referring to the 2 years period December 2020 until November 2022 with one-month interval, of 121 out of the 500 company stocks tracked by S\&P500 (namely, the ones with no missing data in said intervals). Let us call $P_{t,a}$ such cost of an asset $a$, with time index $t$. The return at time step $t$ is defined as
\begin{equation}
    r_{t,a} = \frac{P_{t,a} - P_{t-1,a}}{P_{t-1,a}}
\end{equation}
from which the expected return vector $\tilde{\boldsymbol{\mu}}$ and the covariance matrix $\tilde{\Sigma}$ can be computed as $
    \tilde{\mu}_a = \frac{1}{T}\sum_{t = 1}^T r_{t,a}$ 
and $
    \tilde{\Sigma}_{a,b} = \frac{1}{T-1}\sum_{t = 1}^T (r_{t,a} - \mu_a)(r_{t,b} - \mu_b)$.
We encode the real financial stock market data with decimal precision of $10^{-4}$. 

The number of stocks in an instance is determined by the partition number $w$ \cite{grant2021markovitz}, that describes the granularity of the portfolio discretization. 
The budget is divided in $2^w - 1$ equally large units. 
Each asset decision variable $x_i$ is an integer that can take values from $0$ up to $2^w - 1$, indicating how many of these partitions to allocate towards asset $i$.
Therefore, we need $w$ bits per asset and 
a instance of size $n$ features $N = n /w$ different assets. 
The constraint $\sum_i x_i = 2^w - 1$ ensures that all $2^w-1$ units are invested in one of the stocks.
In the experiments we used a number of bits per asset $w \in \{2,3,5\}$. 
We generate random PO instances by sampling a subset of $N$ assets from the $121$ stocks with complete data uniformly at random.

Notice that $\tilde{\boldsymbol{\mu}}^t \boldsymbol{p}$ is the expected return of a portfolio if $\boldsymbol{p}$ represents the vector of the \emph{portions} of the portfolio for each asset, i.e.\ $0 \le p_i \le 1$ and $\sum_i p_i = 1$. In order to have integer decision variables, the number of chunks $x_i = (2^w-1)p_i$ is used, and in the final formulation (\ref{Markowitz_problem}) of the Markowitz model the factors are absorbed in the objective function, defining $\boldsymbol{\mu} = \tilde{\boldsymbol{\mu}}/(2^w-1)$ and $\Sigma = \tilde{\Sigma}/(2^w-1)^2$.

The last parameter that one needs to set to fully specify the instance is the risk aversion factor $\gamma$, weighting differently the return and the volatility in the objective function. In the experiments we used risk aversion factor $\gamma \in \{0.5,1,2\}$.

\subsection{Experimental implementation.}\label{app:experimental}

The adiabatic theorem \cite{Farhi00, Albash_AQC} states that a quantum system will remain in its instantaneous ground state through small perturbations to its Hamiltonian. Adiabatic quantum computation exploits this fact by preparing a system under the Hamiltonian
\begin{equation}
    H(s)=(1-s) H_0 + s H_P,
\end{equation}\label{adiabatic_hamiltonian}
where $H_0$ is a Hamiltonian with an easy to prepare ground state, $H_P$ encodes the solution of a problem, and the schedule $s$ is evolved from $0$ to $1$. If the evolution meets the conditions of the adiabatic theorem, the system will be at the ground state of $H_P$ at the end of the evolution, hence solving the problem.

A QUBO instance can be mapped into an Ising Hamiltonian by promoting each binary variable $x_i$ into quantum operators. Namely, by substituting $x_i$ for $(1-\sigma^z_i)/2$,
where $\sigma^j_i$ is the Pauli matrix $j$ acting on qubit $i$. Our problem Hamiltonian reads  $H_P=H_M=H_\text{f}+M\, H_\text{c}$, when combining the objective function $f(\vec x)$ and the penality term $M(A\vec x - \vec b)^2$. This way, we recover the diagonal matrix of the QUBO instance. The initial Hamiltonian is usually chosen as $H_0=-\sum_{i=1}^n\sigma^x_i$, as it has an easy to prepare ground state, the equal superposition of computational basis states. This is important, as one of the requirements for the evolution to work is a non-zero overlap between the initial and final ground states.

Deployment of algorithms on available quantum hardware requires precise fine-tuning as well as knowledge of the physical implementations of the device. In this work we target a gate-based ion trap quantum computer, as available through the IonQ cloud service \cite{ionq}. The native interactions available in the device are the following: 
Single qubit gates are fixed $\pi$ and $\pi/2$ rotations along the $X-Y$ plane, with precise control over the relative phase. 
Using this method, rotations around the $Z$ axis are done virtually \cite{mckay2017virtualz}, and incur no noise. 
The IonQ aria-1 device allows for partially-entangling Mølmer-Sørenson \cite{molmer1999msgate, solano1999msgate} gates, that is, a precisely tuned two qubit rotation along the $XX-YY$ plane with virtual control over the relative phases. Since the physics of the ion trap has access to a native $XX$ interaction, we will perform a change of basis to the proposed Hamiltonian, so that the two-body terms in $H_M$ are combinations of $\sigma^x_i\sigma^x_j$, and $H_0$ comprises $\sigma^z_i$ terms. Crucially, the ground state of the initial Hamiltonian is the starting state of the device, resulting in an even easier preparation for the purposes of our evolution.

The adiabatic evolution, ideally performed by slowly sweeping over the interaction parameters of the device, will need to be Trotterized \cite{trotter1959trotter, hatano2005trotter}. By selecting a total evolution time and a discretization step, the evolution can be approximately reproduced by single and two-qubit gates acting on the quantum device. Moreover, this method can be used to control the amount of quantum resources dedicated to solving the problem, which provides an equal starting point to test different QUBO encodings.

In order to conform with the device specifications, and compare different QUBO reformulations under the same conditions, we limit the number of two-qubit gates to $150$. This corresponds to a final annealing time of $100$ with a Trotterization step of $10$ for the instances considered. The parameters used are far from an ideal adiabatic evolution, however, they should still result in an amplification of the ground state of the problem. As shown in Fig.~\ref{Result_plots} (f), this amplification is only significant using the $M_\text{\rm SDP}$ reformulation, making it indispensable even for current noisy devices. This can also hint at advantages on more complex algorithms such as QAOA or VQE when encoding the problem using the proposed $M_\text{\rm SDP}$ reformulation. 

In approximate optimization, outputs that reach a high value for the objective function are desirable even if they do not maximize it. In order to quantify the quality of the solutions, we will use an approximation ratio.
We define the approximation ratio as 
$
    \alpha(\vec x) = (f(\vec x)-f(\vec x_{max}))/(f(\vec x^*)-f(\vec x_{max}))
$
if $\vec x$ satisfies the constraints.

We build the Hamiltonian for the presented problem instances and Trotterize them using the quantum simulation library Qibo \cite{qibo_paper, qibo_software}. Then, the resulting quantum circuits are parsed into native gate instructions for the IonQ \emph{aria-1} device. Code to reproduce this procedure is made available in the following Github repository \cite{github_repo_sergi}. To confirm that this behavior remains consistent as the instances grow, we further simulate exact Trotterized adiabatic evolutions on instances with $12$ and $18$ qubits, see Fig.~\ref{fig:simulation_trotter} and Supplementary Information \ref{app:simulation}. For all models and problem sizes, we consistently observe significantly improved probabilities for the optimal solution and larger approximation ratios using $M_\text{SDP}$ over the baseline with $M_{\ell_1}$.

\subsection{Greedy algorithm for Portfolio Optimization.} \label{app:greedy}
    Any strategy to get a feasible point $\ensuremath\vec x_\text{\rm feas}$ using classical resources is a viable option to obtain $M$ via Eq.~\eqref{eq:MviaBounds}. For various classes of optimization problems, it is possible to apply a greedy heuristic algorithm to efficiently obtain a quasi-optimal point. To exemplify this, we describe a straight-forward greedy strategy  for instances of Portfolio Optimization \eqref{Markowitz_problem}.
    Recall, that given $N$ assets and a partition number $w$, the portfolio is discretized into $2^w-1$ equal fractions.  The following algorithm aims at obtaining solutions by systematically allocating each portfolio portion to the asset that minimizes the objective function when evaluated on the existing segment of the portfolio.   

    \begin{algorithm}
    \caption{\text{GreedyPortfolio}$(\Sigma, \mu, \gamma, N, w)$}\label{greedy_PO}
    \leftskip 5pt 
    \SetKwInOut{Input}{input}
    \SetKwInOut{Output}{output}
    \Input{Risk matrix $\Sigma$, expected return $\mu$, risk aversion factor $\gamma$. \newline Number of assets $N$. \newline Partition number $w$.}
    $\vec x \leftarrow \vec 0 \in \ZZ^N$ \hfill \tcp{Initialize empty portfolio}
    \For{$i \leftarrow 1$ \dots $2^w-1$}{
        \For{$k \leftarrow 1$ \dots $N$}{
        $\vec{\tilde{x}} \leftarrow \vec x$ \\
        $\tilde{x}_k \leftarrow \tilde{x}_k + 1$ \\
        $f' \leftarrow - \vec \mu^t \vec{\tilde{x}} + \gamma \vec{\tilde{x}}^t \Sigma \vec{\tilde{x}}$ \\ 
        \If{$k=1$ or $f' < f^\ast$}{
            $k^{\ast} \leftarrow k$ \\
            $f^{\ast} \leftarrow f'$            
        }
        }
    $x_{k^{\ast}} \leftarrow x_{k^{\ast}} + 1$ \hfill  \tcp{Assign a unit to best asset}
    }
    \textbf{return} optimized portfolio $\vec x$ 
    \end{algorithm}

\section*{Data Availability}
The datasets generated and analysed for the IonQ experiment implementation are available in the quantum-bigm-trotterization repository \url{https://github.com/igres26/quantum-bigM-trotterization/tree/main/data} \cite{github_repo_sergi}.
All the other datasets generated and analysed during the current study are available in the qubo\_mapper repository, \url{https://github.com/EdoardoAlessandroni/qubo\_mapper/tree/master/problems} \cite{github_repo_edo}.

\bibliography{qubo}

\section*{Competing interests}
The Authors declare no competing financial or non-financial interests.

\section*{Author contributions}
EA implemented the numerics. SR conducted the experimental deployment on IonQ.  EA and SR derived the analytical scaling of the gap. IR proved the NP-hardness of finding the optimal big-M.  IR, ET, and LA conceived the project and provided guidance in all steps.  All authors contributed to the conception and write-up of the paper.

\appendix
\renewcommand\thefigure{\thesection.\arabic{figure}}    
\setcounter{figure}{0}   
\onecolumngrid
\clearpage

\section*{Supplementary Information}

\section{Gadgetization: from a general quadratically constrained quadratic optimization problem with integer variables to a linearly-constrained binary quadratic optimization problem}\label{app:gadgets}

We present here the steps of a general procedure, often called \emph{gadgetization}, consisting of elementary operations on the structure of a general combinatorial optimization problem, with the goal of reducing it to a simpler form, namely a quadratic binary problem with linear constraints. Let us consider a problem with quadratic objective function, integer decision variables and quadratic constraints which appear both with an equality and with an inequality condition. Notice that the mentioned formulation can also model polynomial functions, as one can map monomial terms with order greater than two to order two monomials, by adding additional variables to the model. This procedure is similar to what will be shown for the linearization of constraints. Such a problem will have the following form:
\begin{align}
        (P_0) ~~\min_\mathbf{y} \;\;& \mathbf{y}^tQ\mathbf{y} + \mathbf{L}^t\mathbf{y} \\
        \text{s.t.} \; &\mathbf{y}^tq_i\mathbf{y} + \mathbf{l_i}^t\mathbf{y} = b_i\;\; i=1,\dots,m_e \\
        &\mathbf{y}^t\tilde{q}_i\mathbf{y} + \mathbf{\tilde{l}_i}^t\mathbf{y} \ge \tilde{b}_i\;\; i=1,\dots,m_i \label{inequality}  \\
        &\mathbf{y} \in \mathbb{Z}^n  \label{x_integer}\\
        &0 \le y_i \le U_i  \label{x_domain}\;\; i=1,\dots,n
\end{align}
where $m_e$ and $m_i$ are, respectively, the number of equality and inequality constraints and the $n$ integer variables $y_i$ can have values in a finite set as they are upper bounded by some constants $U_i$. Among the model parameters, $Q$, $q_i$ and $\tilde{q}_i$ are $n \times n$ integer valued matrices and $\mathbf{L}$, $\mathbf{l_i}$ and $\mathbf{\tilde{l}_i}$ are $n$-dimensional integer vectors. Notice that integer variables problems are often used to approximate real variables models. In such cases, by increasing the range that the integer variables span, it is possible to reach the desired correspondence between the discretized problem and the continuous one. As an example of this, the discretized Markowitz model in Portfolio Optimization has been analyzed in the present work. 

In the rest of this section we present a four-step procedure that allows one to rewrite a problem in the form $(P_0)$ as a linearly-constrained binary quadratic optimization problem.

\textit{Step 1: rewriting the inequalities as equalities.} In order to deal with equality constraints only, we first need to rewrite the inequalities as equalities, by using additional slack variables $z_i$ that compensate for the deviation between the two sides of the inequality. In this way, Eq.~\eqref{inequality} will become
\begin{align}
        &\mathbf{y}^t\tilde{q}_i \mathbf{y} + \mathbf{\tilde{l}_i}^t\mathbf{y} - z_i = \tilde{b}_i\;\; i=1,\dots,m_i \\
        & 0 \le z_i \le u_i \label{new_inequality} \\ 
        &z_i \in \mathbb{Z}
    \end{align}
where $u_i = \max\{\mathbf{y}^t\tilde{q}_i\mathbf{y} + \mathbf{\tilde{l}_i}^t\mathbf{y} - \tilde{b}_i \; : \eqref{x_integer}, \eqref{x_domain}\}$. Observe that, although Eq.~\eqref{new_inequality} may look like an inequality of the previous kind, it is actually much easier to deal with, since it does not relate variables among themselves. Rather, it only specifies the possible values every single integer variable can assume.

\textit{Step 2: binary expansion of the integer variables.} 
 Since the final formulation will have to contain binary decision variables only, a binary expansion of the integer variables is required, meaning that every integer variable $y_i$ will be replaced by a set of $\lceil log_2(U_i) \rceil$ binary variables $x^i_k$:
\begin{equation}
        y_i = \sum_{k=1}^{\lceil log_2(UB_i)\rceil} \alpha_kx^i_k.
\end{equation}
For the sake of clarity, in the following the superscript $i$ will be dropped from $x_k^i$, as indicating to which integer variable every bit correspond to is not relevant in the present context. As for the expansion coefficients $\alpha_k$ there exist many possible schemes, but here we pick a common one, called \emph{binary encoding}, where $\alpha_k = 2^{k-1}$.

\textit{Step 3: substitution of the quadratic terms in the quadratic equations.}  In order to deal with linear constraints only, we need to get rid of quadratic monomials of the form $x_ix_j$. It is possible to do so by defining a new binary variable 
\begin{equation}
    w_{ij} = x_ix_j  \label{linearization} 
\end{equation}
and substitute every appearance of such monomial with $w_{ij}$ as a new decision variable. To enforce Eq.~\eqref{linearization} it is enough to add to the objective function the term
\begin{equation}
    p(3w_{ij} + x_ix_j - 2x_iw_{ij} - 2x_jw_{ij}). 
\end{equation}
Such a term adds a penalty factor $\ge p$ only when Eq.~\eqref{linearization} is not satisfied, therefore picking $p$ sufficiently large is equivalent to enforcing the substitution and therefore the linearization of the constraints overall.

\textit{Step 4: shifting the linear term in the objective function.}
In a model with only binary variables, it is possible to exploit the equivalence $x_i^2=x_i$ to rewrite the objective function as
\begin{align}
    \mathbf{x}^tQ\mathbf{x} + \mathbf{L}^t\mathbf{x} = 
    \mathbf{x}^t(Q + Diag(L)) \mathbf{x}.
\end{align}
The term $Diag(L)$ indicates a square matrix with all the off-diagonal terms equal to zero, and the diagonal terms equal to $L$. After the application of these four steps, we successfully recover the problem formulation in \eqref{lcBQP}, the starting point in this present work.

\section{SDP Relaxation}
\label{app:sdp}

The following optimization problem 
\begin{align}
     (A_0) ~~&\min_{\mathbf{x}} \mathbf{x}^t Q \mathbf{x} + \mathbf{L}^t \mathbf{x}  \\
     &\text{s.t.} \;  \mathbf{x} \in \{0,1\}^n
\end{align}
can be easily rewritten as
\begin{align}
     \min_{\mathbf{x}} &\Tr\{Y(\mathbf{x})^t\tilde{Q}\}  \\
     \text{s.t.} \; &Y = {1 \choose \mathbf{x}}{1 \choose \mathbf{x}}^t \label{prod betw x}\\
     & \mathbf{x} \in \{0,1\}^n. \label{x binary}
\end{align}
whereby ${1 \choose \mathbf{x}}$ we denote the $n+1$ dimensional column vector with $1$ as the first entry and $\mathbf{x}$ as the remaining $n$ entries, while both $Y$ and $\tilde{Q}$ are $(n+1) \times (n+1)$ matrices and $ Tr\{A^tB\} = \langle A,B\rangle = \sum_{ij} A_{ij}B_{ij}$ denotes the inner product between matrices. In particular, the structure of the objective function is encapsulated in
\begin{equation}
\tilde{Q} = \begin{bmatrix} 
0 & \frac{1}{2}L^t \\
\frac{1}{2}L & Q
\end{bmatrix}.
\end{equation}
The problem can be equivalently reformulated it in the context of convex optimization as
\begin{align}
     (A_1) ~~\min_{Y} &\Tr\{Y^t\tilde{Q}\}  \\
     \text{s.t.} \; &Y \ge 0 \label{Y positive}\\
     &\text{rank}(Y) = 1 \label{Y rank}\\
     & Y_{11} = 1 \label{Y_11=1}\\
     & Y_{1i} =Y_{ii} \;\forall i=2,\dots,n+1 \label{Y binary}
\end{align}
where we replaced condition \eqref{prod betw x} with \eqref{Y positive}, \eqref{Y rank} and \eqref{Y_11=1}, while condition \eqref{x binary} becomes equivalent to \eqref{Y binary}, since $x_i = x_i^2 \iff x_i \in \{0,1\}$.

Problem $(A_1)$ is thus equivalent to $(A_0)$ and its solution space can be viewed as a subset of the positive semi-definite matrices space, rather than the previous vector space.
By removing constraints \eqref{Y rank} and \eqref{Y_11=1} we obtain formulation $(A_2)$, also called \textit{SDP Relaxation}, because it is a relaxation of formulation $(A_0)$ that consists of optimizing over the cone of semidefinite matrices intersected with linear constraints:
\begin{align}
    (A_2) ~~ f^* \coloneqq \min_Y &\Tr\{Y^t\tilde{Q}\}\\
    \text{s.t.} \;  &Y \ge 0\\
    &Y_{1i} =Y_{ii} \quad \forall i=2,\dots,n+1 \\
    &  Y_{ij} \in [0,1] \;\forall i,j=1,\dots,n+1 \label{real_entries}
\end{align}
 In the formulation of $(A_2)$ an additional set of constraints \eqref{real_entries} enforcing to have real entries in the interval $[0,1]$ is added to obtain stronger formulation.
Clearly, solving $(A_2)$ provides a lower bound for $(A_0)$, as $f^* \le \min_{\mathbf{x}} \mathbf{x}^t Q \mathbf{x} + \mathbf{L}^t \mathbf{x}$.

\section{Simulation of Trotterized adiabatic evolution for problem instances of 6, 12 and 18 qubits}\label{app:simulation}

As an extension to the experimental implementation on IonQ's ion-trap quantum computer we present here evidence that the observed behavior is consistent with increased problem size.
Using the same techniques detailed in Sec. \ref{app:experimental} we extend the instances of Trotterized adiabatic evolution to 12 and 18 qubits, as well as include results for LCBOs and SPP. No longer bound by experimental limitations, we can expand the decomposition parameters. In particular, we keep final evolution time of $100$ for 6 qubits, and increase it to $400$ and $1000$ for 12 and 18 qubits respectively for PO and LCBO, and $400$, $800$ and $2000$ for SPP. The Trotterization time-step $\Delta t$ is fixed to $4$ for all instances.

In Fig.~\ref{fig:simulation_trotter} we show these results. We present the probability of measuring the solution after the evolution for reformulations with $M_\text{\rm SDP}$ and with $M_{\ell_1}$. 
We additionally plot the approximation ratio of taking $10^4$ samples on the final state. 
We observe very similar results as seen in Fig.~\ref{Result_plots} (f) and (g) with the increase in problem size and for different benchmarked models. 
Also for higher system sizes, the $M_\text{\rm SDP}$ reformulation 
yields a significant probability of observing the optimal result, while the $M_{\ell_1}$ reformulations yield comparable results to a randomized searched.

\end{document}